\documentclass{article}
\usepackage{geometry}
\usepackage{amsmath}
\usepackage{amsfonts}
\usepackage{amstext}
\usepackage{amssymb}
\usepackage{amsthm}
\usepackage{eufrak}
\newtheorem{lemma}{Lemma}

\newtheorem{theorem}{Theorem}

\theoremstyle{definition}
\newtheorem{definition}{Definition}

\theoremstyle{remark}
\newtheorem{remark}{Remark}
\newtheorem{example}{Example}
\newcommand{\cm}{{\EuFrak B}}

\usepackage{color}
\usepackage{graphicx}

\title{Implied volatility of basket options at extreme strikes}
\author{Archil Gulisashvili\footnote{Department of Mathematics, Ohio
  University, Athens, Ohio, USA. Email: \texttt{gulisash@ohio.edu}}\quad
 and\quad Peter Tankov\footnote{Laboratoire de Probabilit\'es et Mod\`eles Al\'eatoires, Universit\'e Paris Diderot, Paris, France and International Laboratory of Quantitative Finance,
  National Research University ``Higher School of Economics", Moscow,
  Russia. Email: \texttt{tankov@math.univ-paris-diderot.fr}}}
\date{}
\begin{document}
\maketitle

\begin{abstract}
In the paper, we characterize the asymptotic behavior of the implied volatility of a
basket call option at large and small strikes in a variety of
settings with increasing generality. First, we obtain an asymptotic formula with an error bound
for the left wing of the implied volatility, under the assumption that the dynamics of asset prices are described by the multidimensional Black-Scholes model. Next, we find the leading term of asymptotics of the implied volatility in the case where the asset prices follow the multidimensional Black-Scholes model with time change by
an independent increasing stochastic process. Finally, we deal with a general
situation in which the dependence between the assets is described by a given copula
function. In this setting, we obtain a model-free tail-wing formula
that links the implied volatility to a special characteristic of the copula called \emph{the weak lower tail dependence function}.  
\end{abstract}

\section{Introduction}
In option markets, prices of vanilla call and put options are commonly
quoted in terms of their
\emph{implied volatility} $I(T,K)$, defined as the value of the
volatility parameter which must be substituted into the Black-Scholes
option pricing formula to obtain the quoted option price. Similarly,
given a risk-neutral model, one can define the function $(T,K)\mapsto
I(T,K)$ from the prices of vanilla options computed for that
model. However, since in most stochastic asset price models the implied volatility 
function is not known explicitly, it becomes important to obtain efficient and accurate asymptotic
approximations for it. Such approximations are useful for at least
two reasons. First, they may shed light on the qualitative behavior
of the implied volatility in the asset price model, and also on the effect of different
model parameters on the shape of the model-generated implied
volatility surface. Second, they allow to perform an approximate
calibration of the model by comparing the market implied volatility with the asymptotic
approximation. Such preliminary estimates can be used as intelligent
guesses in the construction of a numerical calibration algorithm to accelerate its
convergence. 

Approximations to the implied volatility have been studied
by many authors in a variety of asymptotic regimes, both in specific
models and in model-independent settings. One of the early references
on the subject is the book by Lewis \cite{lewis} dealing with
stochastic volatility models. Various model-free formulas describing the wing behavior of the implied volatility were obtained in the last decade.
To our knowledge, celebrated Lee's moment formulas were the first model-independent asymptotic formulas for the implied volatility at extreme strikes (see \cite{lee.04}). Lee's results were later refined by Benaim
and Friz \cite{benaim.friz.09,benaim.friz.08} and Gulisashvili \cite{gulisashvili.10,gulisashvili.12a,gulisashvili.12b}. In Gao and Lee \cite{gao.lee.11}, higher order asymptotic formulas for the implied volatility at extreme strikes were found, and in Tehranchi \cite{tehranchi.14}, uniform estimates for the implied volatility are obtained. Small-time behavior of
implied volatility is analyzed, among other papers, in \cite{busca02} (in local volatility
models), \cite{forde2009small} (for the Heston stochastic volatility
model), \cite{medvedev04} (for jump-diffusions), and in \cite{andersen2013asymptotics,Ford_Figueroa-Lopez,mijatovic2013new,tankov.pplnmf} (for exponential L\'evy models). Formulae for the implied volatility far
from maturity are given in \cite{jacquier.forde.11} (for the Heston
model) and \cite{tehranchi.09}
(model-independent). Finally, sharp
price and implied volatility
approximations for various models have been obtained as ``expansions
around the Black-Scholes model'' in \cite{benhamou.al.09,gobet2010time}. 

Implied volatility is also quoted in the market for options on a basket of
stocks (or on a market index). Note that the Black-Scholes formula can
be applied to price a vanilla option by considering the entire basket (index) as a
log-normal random variable. In such a case, finding reliable asymptotic approximations to the implied volatility
can be even more important, since calculating the exact value numerically can be computationally very expensive due
to the large dimension of the basket. Approximations based on the
\emph{small-noise} asymptotics in multidimensional local volatility
models have been developed in \cite{avellaneda.al.02} and more
recently refined in \cite{bayer2013asymptotics}, but in other
asymptotic regimes, much less is known about multi-asset options, than
in the single-asset case. 

Our main goal in the present paper is to characterize the asymptotic behavior of
the implied volatility of a call option on a basket of stocks (with
positive weights) for large and small strikes. Three
different classes of multidimensional risk-neutral models with
increasing generality are considered in the paper. In Section
\ref{S:nBS}, we discuss the case of correlated log-normal
assets, in other words, the assets which follow the multidimensional
Black-Scholes model. Using a recent characterization of the tail behavior of sums
of correlated log-normal random variables
\cite{gulisashvili.tankov.13}, we obtain a sharp asymptotic formula with error
estimates for the implied volatility at small strikes. On the other hand, the asymptotics of the 
implied volatility at large strikes 
can be easily characterized using the results obtained in
\cite{asmussen2008asymptotics}. It turns out that for large strikes, 
the implied volatility of a basket call option is approximated by the highest volatility among the stocks in the basket. 

Section \ref{S:timechange} deals with the case, where the assets
follow the multidimensional Black-Scholes model time-changed by an
independent increasing stochastic process. It is assumed in this section that the marginal density of the
time-change process decays at infinity like the function $s\mapsto s^{\alpha}e^{-\theta s}$ with $\alpha\in\mathbb{R}$
and $\theta> 0$. The class of such models includes 
standard multidimensional extensions of various exponential L\'evy models, for instance,
of the variance gamma model, the normal inverse Gaussian model, or the generalized hyperbolic model. To our knowledge, for such a class of multidimensional models, the tail behavior of marginal distributions has not been studied before. In Section \ref{S:timechange}, we provide two-sided estimates for the distribution
function of the asset price in the time-changed multidimensional Black-Scholes model, and use these estimates to find the leading term in the asymptotic expansion of the
implied volatility. 

Finally, in section \ref{S:copula}, we deal with the case
where the assets in the basket are correlated, and the dependence structure is described by a given copula function. Here we obtain an asymptotic formula that can be considered as a generalization to the multidimensional setting of one of the tail-wing formulae established in \cite{benaim.friz.09}. The new tail-wing formula uses a special characteristic of the copula called \emph{weak lower tail dependence function}. This notion was recently introduced in \cite{tankov.14}.  

\paragraph{Remarks on the notation used in the paper}
\begin{itemize}
\item Let $f$ and $g$ be functions defined on $\mathbb{R}$, and let $a\in[-\infty,\infty]$. Throughout the present paper, we write ``$f\sim g$ as $x\to a$" provided that
$$
\lim_{x\to a}\frac{f(x)}{g(x)} = 1.
$$ 
We also use the notation ``$f \lesssim g$ as $x\to a$" if 
$$
\limsup_{x\to a}\frac{f(x)}{g(x)} \leq 1,
$$ 
and write ``$f(x)\approx g(x)$ as $x\to a$" if there exist $c_1> 0$ and $c_2> 0$ such that 
$$
c_1 g(x)\le f(x)
\le c_2 g(x)
$$ 
for all $x$ in some neighborhood of $a$.

\item A positive function $f$ defined in $[a,\infty)$ for some $a> 0$ is called regularly varying at infinity with index $\alpha\in\mathbb{R}$ if for any $\lambda> 0$,
$$
\lim_{x\to 0}\frac{f(\lambda x)}{f(x)} =\lambda^{\alpha}.
$$ 
for all $\alpha >0$. The class of all regularly varying functions with index $\alpha$ is denoted by $R_{\alpha}$.
The elements of the class $R_0$ are called slowly varying functions. Regularly varying functions at zero can be defined similarly.

\item The following set will be used in the paper:
\begin{align*}
\Delta_d:&= \{w\in \mathbb R^d: w_i\geq 0,i=1,\dots,d, \text{and}\ \sum_{i=1}^d
w_i = 1\}.
\end{align*}

\item Let $w \in \Delta_d$. We set
\begin{equation}
\mathcal E(w):= -\sum_{i=1}^d w_i \log w_i,
\label{E:sas}
\end{equation}
with the convention $x \log x = 0$ for $x=0$.
\end{itemize}

\section{Model-free formulae for the implied volatility}\label{S:modelfree}
Let $X_t$ be a non-negative martingale on a filtered probability space $(\Omega,\mathcal{F},\left\{\mathcal{F}_t\right\}_{t\ge 0},\mathbb{P})$. Consider a stochastic model where the process $X$ models the price dynamics of an asset. Define the call and put pricing functions in the price model described above by
\begin{equation}
C(T,K)=\mathbb{E}[(X_T-K)^{+}]\quad\mbox{and}\quad P(T,K)=\mathbb{E}[(K-X_T)^{+}],
\label{E:putp}
\end{equation}
respectively. Here $T> 0$ is the maturity, while $K> 0$ is the strike price.

The implied volatility $(T,K)\mapsto I(T,K)$ is determined from the following equality:
$$
C(K,T)=C_{BS}(T,K,\sigma=I(T,K)),
$$ 
where the symbol $C_{BS}$ stands for the Black-Scholes call pricing function.
In the sequel, the maturity $T$ will be fixed, and the implied volatility will be considered as a function of only the strike price.

We will next formulate two model-free asymptotic formulas, 
characterizing the left-wing behavior of the implied volatility in terms of the put pricing function. These formulas will be needed below. Suppose the initial condition for the price process is $X_0=1$. Suppose also that the asset price model does not have atoms at zero. The previous assumption means that $\mathbb{P}(X_T=0)=0$. Then the following asymptotic formula (a zero order formula for the implied volatility) holds:
\begin{align}
I(K)&=\frac{\sqrt{2}}{\sqrt{T}}\sqrt{\log\frac{1}{\widetilde{P}(K)}-\frac{1}{2}\log\log\frac{K}
{\widetilde{P}(K)}}
-\frac{\sqrt{2}}{\sqrt{T}}\sqrt{\log\frac{K}{\widetilde{P}(K)}-\frac{1}{2}\log\log\frac{K}
{\widetilde{P}(K)}} \nonumber \\
&+O\left(\left(\log\frac{K}{\widetilde{P}(K)}\right)^{-\frac{1}{2}}\right)
\label{E:assy8}
\end{align}
as $K\rightarrow 0$. Here $\widetilde{P}$ is a positive function 
satisfying the condition $P(K)\approx\widetilde{P}(K)$ as $K\rightarrow 0$. Formula (\ref{E:assy8}) was established in \cite{gulisashvili.10} (see also Theorem 9.29 in \cite{gulisashvili.12b}). The fact that the absence of atoms is a necessary condition for the validity of formula
(\ref{E:assy8}) was noticed in \cite{de2013shapes} (see also
\cite{gulisashvili.13}).  

The next asymptotic formula (a first-order formula for the implied volatility) 
can be easily deduced from the results formulated in \cite[Sections 9.6 and 9.9]{gulisashvili.12b}:
\begin{align}
I(K)&=\frac{\sqrt{2}}{\sqrt{T}}\sqrt{
\log\frac{1}{P(K)}-\frac{1}{2}\log\log\frac{K}{P(K)}
+\log B(K)} \nonumber \\
&\quad-\frac{\sqrt{2}}{\sqrt{T}}\sqrt{\log\frac{K}{P(K)}
-\frac{1}{2}\log\log\frac{K}{P(K)}
+\log B(K)} \nonumber \\
&\quad+O\left(\log\log\frac{K}{P(K)}\left(\log\frac{K}{P(K)}\right)^{-\frac{3}{2}}\right)
\label{E:imp6}
\end{align}
as $K\rightarrow 0$,
where
\begin{equation}
B(K)=\frac{\sqrt{\log\frac{1}{P(K)}}
-\sqrt{\log\frac{K}{P(K)}}}
{2\sqrt{\pi}\sqrt{\log\frac{1}{P(K)}}}.
\label{E:B}
\end{equation}
Formula (\ref{E:imp6}) takes into account the results obtained in
\cite{gao.lee.11}. It provides more terms in the asymptotic expansion of the implied volatility at small strikes than
formula (\ref{E:assy8}) with $\widetilde{P}=P$. More information on model free formulas for the implied volatility can be found in \cite{gulisashvili.12b}.

\section{Basket options in multidimensional Black-Scholes model}\label{S:nBS}
Our goal in the present section is to characterize the asymptotic behavior of the implied volatility
at small strikes in the case of a basket option of European style in the n-dimensional driftless Black-Scholes model. We assume that the interest rate is equal to zero. Let $S^1,\dots,S^n$ be a basket of assets such that
$$
\log\widetilde{S}_t=\log\widetilde{S}_0-\frac{\rm diag(\cm)t}{2}+\cm^{\frac{1}{2}}W_t,
$$
where $\widetilde{S}_t=(S^1_t,\dots,S^n_t)$, $\widetilde{S}_0=(S^1_0,\dots,S^n_0)$, $W$ is an n-dimensional standard Brownian motion, $\cm$ is the covariance matrix,
and $\rm diag(\cm)$ stands for the main diagonal of $\cm$. We denote by $(\lambda_1,\dots,\lambda_n)\in \Delta_n$
the weight vector associated with the assets in the basket. 

Consider the price process of the following form:
\begin{equation}
S_t = \sum_{i=1}^n \lambda_i S^i_t,\quad t\ge 0.
\label{E:pricep}
\end{equation}
The initial condition for the process $S$ is given by
$S_0= \sum_{i=1}^n \lambda_i S^i_0$, and we will assume in the sequel that $S^i_0=1$ for all $1\le i\le n$. The previous condition implies that $S_0=1$.
Therefore,
$S_t=\sum_{i=1}^n\exp\{Y_t^i\}$,
where
\begin{equation}
Y_t^i=\log\lambda_i-\frac{b_{ii}t}{2}+\sum_{j=1}^n\beta_{ij}W_t^j,\quad 1\le i\le n.
\label{E:pr1}
\end{equation}
In (\ref{E:pr1}), the symbols $\beta_{ij}$ stand for the elements of the matrix $\cm^{\frac{1}{2}}$. We also set
\begin{equation}
\mu_{i,t}=\log\lambda_i-\frac{b_{ii}t}{2},\quad 1\le i\le n.
\label{E:weights}
\end{equation}
It is clear that the following equality holds: $\exp\{Y_t^i\}=\lambda_iS_t^i$, $t> 0$, $1\le i\le n$.
\subsection{Asymptotics of put pricing functions in multidimensional
  Black-Scholes model}\label{asBS}
The
distribution density of the random variable $S_T$ will be denoted by
$p_T$. An asymptotic formula for $p_T$ was recently established in
\cite{gulisashvili.tankov.13}. Let us briefly recall the notation used
in that paper. Let $\bar w\in \Delta_n$ be the unique vector such that 
\begin{align}
\bar w^\perp \cm \bar w = \min_{w\in \Delta_n} w^\perp \cm w. \label{minprob}
\end{align}
The existence and uniqueness of $\bar w$ follows from the non-degeneracy of the matrix
$\cm$. 
We
let 
\begin{align}
\bar n:= \text{Card}\,\{i=1,\dots,n: \bar w_i \neq 0\},\label{nbar.eq}
\end{align}
$$
\bar I:= \{i=1,\dots,n: \bar w_i \neq 0\}:= \{\bar k(1),\dots,\bar
k(\bar n)\},
$$ 
$\bar \mu\in \mathbb R^{\bar n}$ with $\bar \mu_i = \mu_{\bar k(i)}$,
and $\bar \cm \in M_{\bar n}(\mathbb R)$ with $\bar \cm_{ij} = \cm_{\bar
k(i),\bar k(j)}$. The inverse matrix of $\bar\cm$ is denoted by $\bar
\cm^{-1}$  and its elements and row sums by $\bar a_{ij}$ and $\bar A_k := \sum_{j=1}^{\bar n}
\bar a_{kj}$. 
We refer the interested reader to
\cite{gulisashvili.tankov.13} for more details and explanations on
this notation. 

It was established in
\cite{gulisashvili.tankov.13} that under a special restriction on the correlation matrix $\cm$ (Assumption $(\mathcal A)$ in \cite{gulisashvili.tankov.13}), the following asymptotic formula is valid as $x\to 0$: 
\begin{align}
p_T(x)& =  C_T\left(\log\frac{1}{x}\right)^{\frac{1-\bar{n}}{2}}
x^{-1+\frac{1}{T}\sum_{k=1}^{\bar{n}}
\bar{A}_k\left(\log\frac{\bar{A}_1+\cdots+\bar{A}_{\bar{n}}}{\bar A_k}+\bar{\mu}_{k,T}\right)} \nonumber \\
&\quad\exp\left\{-\frac{1}{2T}(\bar A_1+\cdots+\bar A_{\bar{n}})\log^2\frac{1}{x}\right\} \left(1+O\left(\left(\log
      \frac{1}{x} \right)^{-1}\right)\right),
      \label{E:forr1}
\end{align}
where the constant $C$ is given by 
\begin{align}
&C_T=\frac{1}{\sqrt{2\pi T}\sqrt{\left|\bar\cm\right|}}\frac{\sqrt{\bar
    A_1+\cdots+\bar A_{\bar n}}}{\sqrt{\bar A_1\cdots \bar A_{\bar n}}}
\nonumber \\
&\exp\left\{-\frac{1}{2T}\sum_{ i,j=1}^{\bar n} \bar
  a_{ij}\left(\log\frac{\bar A_1+\cdots+\bar A_{\bar n}}{\bar
      A_i}+\bar \mu_{i,T}\right)
\left(\log\frac{\bar A_1+\cdots+\bar A_{\bar n}}{\bar A_j}+\bar \mu_{j,T}\right)\right\}.
\label{E:for1}
\end{align}
Using formula (\ref{E:forr1}), we can characterize the asymptotic behavior of the put pricing function $P$ at small strikes. This can be done as follows. Consider the fractional integral of order two defined by
\begin{equation}
F_2M(\sigma)=\int_{\sigma}^{\infty}(\tau-\sigma)M(\tau)d\tau,
\label{E:A}
\end{equation}
where $M$ is a positive function on $(0,\infty)$. Since
$$
P(K)=\int_0^K(K-x)p_T(x)dx,
$$
it is not hard to see that
\begin{equation}
P(K)=S^{-1}F_2M(S),\quad\mbox{where}\quad S=K^{-1}\quad\mbox{and}\quad M(y)=y^{-3}
p_T\left(y^{-1}\right).
\label{E:iv1}
\end{equation}
Using (\ref{E:forr1}), we get
\begin{align}
&M(y)=C_T\left(\log y\right)^{\frac{1-\bar{n}}{2}}
y^{-2-T^{-1}\sum_{k=1}^{\bar{n}}\bar{A}_k\left(\log\frac{\bar{A}_1+\cdots+\bar{A}_{n}}
{\bar{A}_k}+\mu_{k,T}\right)} \nonumber \\
&\quad\exp\left\{-\frac{1}{2T}(\bar{A}_1+\cdots+\bar{A}_{\bar n})\log^2y\right\} 
\left(1+O\left(\left(\log y\right)^{-1}\right)\right)
\label{E:ii}
\end{align}
as $y\rightarrow\infty$, where $C_T$ is given by (\ref{E:for1}).

In \cite{gulisashvili2010asymptotic}, a general asymptotic formula was obtained for fractional integrals (see also Theorem 5.3 in \cite{gulisashvili.12b}).
We will next formulate this general result. Suppose 
$$
M(y)=a(y)e^{-b(y)}\quad\mbox{for all}\quad y\ge c
$$
where $c> 0$ is some number. Suppose also that the following conditions hold:
\begin{enumerate}
\item 
$y|a^{\prime}(y)|\le\gamma a(y)$ for some $\gamma> 0$ and all $y> c$.
\item 
$b(y)=B(\log y)$, where $B$ is a positive increasing function on $(c,\infty)$
such that $B^{\prime\prime}(y)\approx 1$ as $y\rightarrow\infty$.
\end{enumerate}
Then as $\sigma\rightarrow\infty$,
\begin{equation}
F_2M(\sigma)=\frac{M(\sigma)}{b^{\prime}(\sigma)^2}(1+O((\log\sigma)^{-1})).
\label{E:iii}
\end{equation}

The function $M$ in (\ref{E:iv1}) satisfies the conditions formulated above. Next, using (\ref{E:iv1}), 
(\ref{E:ii}), and (\ref{E:iii}) with 
$$
B(u)=\frac{1}{2T}(\bar{A}_1+\cdots+\bar{A}_{\bar n})u^2,
$$
we establish the following assertion.
\begin{theorem}\label{P:eshchio}
Let $P$ be the price of the put option defined in (\ref{E:putp}), and suppose Assumption ($\mathcal{A}$) holds for the covariance matrix 
$\cm$ (see \cite{gulisashvili.tankov.13}). Then, as $K\rightarrow 0$,
\begin{align}
P(K)&=\delta_0
\left[\log\frac{1}{K}\right]^{\delta_1}\left(\frac{1}{K}\right)^{\delta_2}
\exp\left\{-\delta_3\log^2\frac{1}{K}\right\}
\left(1+O\left(\left(\log\frac{1}{K}\right)^{-1}\right)\right),
\label{E:imp1}
\end{align}
where
$$
\delta_0=\frac{C_TT^2}{\left(\bar{A}_1+\cdots+\bar{A}_{\bar{n}}\right)^2},
\quad\delta_1=-\frac{3+\bar{n}}{2},
$$
$$
\delta_2=-1-\frac{1}{T}\sum_{k=1}^{\bar{n}}\bar{A}_k\left(\log\frac{\bar{A}_1+\cdots+\bar{A}_{\bar{n}}}{\bar{A}_k}+\mu_{k,T}\right),
\quad\delta_3=\frac{1}{2T}(\bar{A}_1+\cdots+\bar{A}_{\bar{n}}),
$$
and $C_T$ is given by (\ref{E:for1}).
\end{theorem}

Formula (\ref{E:imp1}) will be used in the next subsection to characterize the left-wing behavior of the implied volatility associated with a basket option in the multidimensional Black-Scholes model.
\subsection{Left-wing asymptotic behavior of the implied volatility associated with basket options}\label{S:leftw}
The next statement characterizes the asymptotic behavior of the implied volatility for small strikes. 
\begin{theorem}\label{P:esh}
Suppose Assumption $(\mathcal{A})$ holds for the covariance matrix 
$\cm$. Then, as $K\rightarrow 0$,
\begin{align}
&I(K)=\frac{1}{\sqrt{\bar{A}_1+\cdots+\bar{A}_{\bar{n}}}}-\frac{2\sum_{k=1}^{\bar{n}}\bar{A}_k
\left(\log\frac{\bar{A}_1+\cdots+\bar{A}_{\bar{n}}}{\bar{A}_k}+\mu_{k,T}\right)+T}
{2(\bar{A}_1+\cdots+\bar{A}_{\bar{n}})^{\frac{3}{2}}}\left(\log\frac{1}{K}\right)^{-1}
\nonumber \\
&\quad-\frac{T(\bar{n}-1)}{2(\bar{A}_1+\cdots+\bar{A}_{\bar{n}})
^{\frac{3}{2}}}\log\log\frac{1}{K}\left(\log\frac{1}{K}\right)^{-2}+O\left(\left(\log\frac{1}{K}\right)^{-2}\right).
\label{E:seccund}
\end{align}
\end{theorem}
\begin{remark}
The leading term in the implied volatility expression above can also
be written as
$$
\lim_{K\downarrow 0} I(K) =
\frac{1}{\sqrt{\bar{A}_1+\cdots+\bar{A}_{\bar{n}}}} = \sqrt{\min_{w\in
  \Delta _n} w^\perp \cm w}. 
$$
\end{remark}
\begin{proof}
It follows from (\ref{E:imp1}) that as $K\rightarrow 0$,
\begin{align}
\log\frac{1}{P(K)}&=\log\frac{1}{\delta_0}-\delta_1\log\log\frac{1}{K}
-\delta_2\log\frac{1}{K}+\delta_3\log^2\frac{1}{K} \nonumber \\
&+O\left(\left(\log\frac{1}{K}\right)^{-1}\right)
\label{E:sor1}
\end{align}
and
\begin{align}
\log\frac{K}{P(K)}&=\log\frac{1}{\delta_0}
-\delta_1\log\log\frac{1}{K}
-(\delta_2+1)\log\frac{1}{K} \nonumber \\
&+\delta_3\log^2\frac{1}{K}+O\left(\left(\log\frac{1}{K}\right)^{-1}\right)
\label{E:sor2}
\end{align}
where $\delta_0$, $\delta_1$, $\delta_2$, and $\delta_3$ are such as in Theorem 
\ref{P:eshchio}. Moreover, the error 
term in (\ref{E:imp6}) can be represented as follows:
\begin{equation}
O\left(\log\log\frac{1}{K}\left(\log\frac{1}{K}\right)^{-3}\right).
\label{E:sor33}
\end{equation}

We will next characterize the asymptotic behavior of $\log B(K)$ as $K\rightarrow 0$. Denote the functions
on the right-hand side of (\ref{E:sor1}) and (\ref{E:sor2}) by $V_1(K)$ and $V_2(K)$, respectively.
Then, using (\ref{E:B}), (\ref{E:sor1}), and (\ref{E:sor2}), we obtain
$$
\log B(K)=\log\frac{1}{2\sqrt{\pi}}+\log\left[1-\sqrt{1-\frac{V_1(K)-V_2(K)}{V_1(K)}}\right].
$$

It is easy to see that $\log(1-\sqrt{1-h})=\log\frac{h}{2}+O(h)$ as $h\rightarrow 0$. Put 
$h=\frac{V_1(K)-V_2(K)}{V_1(K)}$. Then we have
$$
\log B(K)=\log\frac{1}{2\sqrt{\pi}}+\log\frac{V_1(K)-V_2(K)}{2V_1(K)}
+O\left(\left(\log\frac{1}{K}\right)^{-1}\right), 
$$
and hence
\begin{equation}
\log B(K)=\log\frac{1}{4\sqrt{\pi}\delta_3}-\log\log\frac{1}{K}
+O\left(\left(\log\frac{1}{K}\right)^{-1}\right)
\label{E:sor3}
\end{equation}
as $K\rightarrow 0$.

Our next goal is to simplify formula (\ref{E:imp6}) by taking into account (\ref{E:sor1}), (\ref{E:sor2}),
and (\ref{E:sor3}), and replacing the error term by the expression in (\ref{E:sor33}). We can drop the terms 
$O\left(\left(\log\frac{1}{K}\right)^{-1}\right)$ in (\ref{E:sor1}), (\ref{E:sor2}), and (\ref{E:sor3}), using the mean value
theorem. This will introduce an error term $O\left(\left(\log\frac{1}{K}\right)^{-2}\right)$ in
the formula that follows from formula (\ref{E:imp6}). Thus
\begin{align}
&I(K)
=\frac{\sqrt{2}}{\sqrt{T}}
\sqrt{\widetilde{V}_1(K)-\frac{1}{2}\log\widetilde{V}_2(K)
+\log\frac{1}{4\sqrt{\pi}\delta_3}-\log\log\frac{1}{K}} \nonumber \\
&\quad-\frac{\sqrt{2}}{\sqrt{T}}\sqrt{\widetilde{V}_2(K)
-\frac{1}{2}\log\widetilde{V}_2(K)
+\log\frac{1}{4\sqrt{\pi}\delta_3}-\log\log\frac{1}{K}} \nonumber \\
&\quad+O\left(\left(\log\frac{1}{K}\right)^{-2}\right)
\label{E:karr}
\end{align}
as $K\rightarrow 0$, where $\widetilde{V}_1(K)$ and $\widetilde{V}_2(K)$ denote the functions on the right-hand side of (\ref{E:sor1}) and (\ref{E:sor2}), respectively, without the terms 
$O\left(\left(\log\frac{1}{K}\right)^{-1}\right)$. Next, using the mean value theorem, we see that
it is possible to replace $\widetilde{V}_2(K)$ in the expression $\log\widetilde{V}_2(K)$ in
formula (\ref{E:karr}) by $\delta_3\log^2K$. Now, taking into account the definitions of $\widetilde{V}_1(K)$ and $\widetilde{V}_2(K)$,
we obtain
\begin{align}
&I(K)=\frac{\sqrt{2}}{\sqrt{T}}\sqrt{-\log\left[4\sqrt{\pi}\delta_0
\delta_3^{\frac{3}{2}}\right]-(\delta_1+2)\log\log\frac{1}{K}-\delta_2\log\frac{1}{K}
+\delta_3\log^2\frac{1}{K}} \nonumber \\
&-\frac{\sqrt{2}}{\sqrt{T}}\sqrt{-\log\left[4\sqrt{\pi}\delta_0
\delta_3^{\frac{3}{2}}\right]-(\delta_1+2)\log\log\frac{1}{K}-(\delta_2+1)\log\frac{1}{K}
+\delta_3\log^2\frac{1}{K}} \nonumber \\
&+O\left(\left(\log\frac{1}{K}\right)^{-2}\right)
\label{E:zabyl}
\end{align}
as $K\rightarrow 0$. Put 
$$
h_1(K)=\frac{-\log\left[4\sqrt{\pi}\delta_0
\delta_3^{\frac{3}{2}}\right]-(\delta_1+2)\log\log\frac{1}{K}-\delta_2\log\frac{1}{K}}
{\delta_3\log^2\frac{1}{K}}
$$
and
$$
h_2(K)=\frac{-\log\left[4\sqrt{\pi}\delta_0
\delta_3^{\frac{3}{2}}\right]-(\delta_1+2)\log\log\frac{1}{K}-(\delta_2+1)\log\frac{1}{K}}
{\delta_3\log^2\frac{1}{K}}.
$$
It follows from (\ref{E:zabyl}) that
\begin{align}
I(K)&=\frac{\sqrt{2}\sqrt{\delta_3}}{\sqrt{T}}\log\frac{1}{K}
\left[\sqrt{1+h_1(K)}-\sqrt{1+h_2(K)}\right] 
+O\left(\left(\log\frac{1}{K}\right)^{-2}\right)
\label{E:secco}
\end{align}
as $K\rightarrow 0$. Next, using the formula 
$\sqrt{1+h}=1+\frac{1}{2}h-\frac{1}{8}h^2+O(h^3)$
as $h\rightarrow 0$ in (\ref{E:secco}), we get
\begin{align}
&I(K)=\frac{1}{\sqrt{2T\delta_3}}
+\frac{1+2\delta_2}{4\delta_3\sqrt{2T\delta_3}}
\left(\log\frac{1}{K}\right)^{-1} \nonumber +\frac{\delta_1+2}{2\delta_3\sqrt{2T\delta_3}}\log\log\frac{1}{K}\left(\log\frac{1}{K}\right)^{-2} 
\nonumber \\
&+O\left(\left(\log\frac{1}{K}\right)^{-2}\right)
\label{E:seccond}
\end{align}
as $K\rightarrow 0$. Finally, plugging the values of $\delta_1$, $\delta_2$, and $\delta_3$ given
in Theorem \ref{P:eshchio} into formula (\ref{E:seccond}), we obtain formula (\ref{E:seccund}).

This completes the proof of Theorem \ref{P:esh}.
\end{proof}

\begin{remark}[Implied volatility in the multidimensional
  Black-Scholes model for large strikes]
From Theorem 1 in \cite{asmussen2008asymptotics}, it follows that 
$$
\mathbb P[S_t \geq K] \sim \frac{m_n\sigma\sqrt{t}}{\sqrt{2\pi} \log K}
\exp\left\{-\frac{(\log K - \mu)^2}{2\sigma^2 t}\right\},\quad K\to \infty,
$$
where $\sigma^2 = \max_{k=1,\dots,n} \cm_{kk}$, $\mu = \max_{\mu_{k,t}:
  \cm_{kk} = \sigma^2}$ and $m_n = \#\{k : \cm_{kk} = \sigma^2, \mu_{k,t}
= \mu\}$. From this result, we easily deduce that 
$$
\mathbb E[(S_t-K)^+]\approx \frac{K}{\log^2 K} \exp\left\{-\frac{(\log K - \mu)^2}{2\sigma^2 t}\right\},\quad K\to \infty.
$$
Applying Corollary 2.4 in \cite{gulisashvili.10} (which is nothing but
the
right-tail version of formula \eqref{E:assy8}), we conclude that 
$$
I(K) = \sigma + O\left(\frac{\psi(K)}{\log K}\right)
$$
as $K \to +\infty$, where $\psi$ is any function satisfying
$\psi(K)\to +\infty$ as $K\to +\infty$. 
\end{remark}
\subsection{The case where $n=2$}\label{SS:n2}
The detailed discussion of the behavior of the distribution of the sum
of two log-normal variables can be found in \cite{gao2009asymptotic}
  and \cite{gulisashvili.tankov.13}. The covariance matrix in this case is as follows: $\cm=[b_{ij}]$, where
$b_{11}=\sigma_1^2$, $b_{12}=b_{21}=\rho\sigma_1\sigma_2$, $b_{22}=\sigma_2^2$ with $\sigma_1> 0$, $\sigma_2> 0$, and 
the correlation coefficient satisfies $-1<\rho< 1$. We will also assume $\sigma_1\geq
\sigma_2$. 
Note that the case where $\rho<\frac{\sigma_2}{\sigma_1}$ is a regular case, and Assumption
($\mathcal{A}$) holds. In the case where $\rho>\frac{\sigma_2}{\sigma_1}$, we have to rearrange the rows and the columns of $\cm$ (see the example in Section 2.1 of \cite{gulisashvili.tankov.13}). Then $\bar\cm=(\sigma_2^2)$, and Assumption ($\mathcal{A}$) holds. The case where $\rho=\frac{\sigma_2}{\sigma_1}$ is exceptional. Here Assumption ($\mathcal{A}$) does not hold. 

The following asymptotic formulas for the implied volatility follow from \eqref{E:seccund}:
\begin{itemize}
\item Suppose $\rho > \frac{\sigma_2}{\sigma_1}$. Then 
\begin{equation}
I(K) = \sigma_2 -
  \sigma_2\log\lambda_2
  \left(\log\frac{1}{K}\right)^{-1} + O\left(\left(\log\frac{1}{K}\right)^{-2}\right)
\label{E:assy1}
\end{equation}
as $K\rightarrow 0$.
\item Suppose $\rho <\frac{\sigma_2}{\sigma_1}$. Then
\begin{align}
I(K) =&\sigma_\infty -\sigma_\infty \Bigg(\frac{T}{2}\sigma_\infty^2 +\left[\log\lambda_1 - 
  \frac{\sigma^2_1 T}{2} - \log \bar v\right]\bar v \nonumber \\
  &+\left[\log\lambda_2 - \frac{\sigma^2_2 T}{2} - \log(1-\bar v)\right](1-\bar v) \Bigg)\left(\log \frac{1}{K}\right)^{-1} \nonumber \\
& - \frac{T}{2}\sigma_\infty^3 \frac{\log\log \frac{1}{K}}{\log^2
  \frac{1}{K}} + O\left(\left(\log\frac{1}{K}\right)^{-2}\right)
  \label{E:assy2}
\end{align}
as $K\rightarrow 0$, where 
$$
\sigma_\infty = \frac{\sigma_1 \sigma_2 \sqrt{1- \rho^2} }{\sqrt{\sigma_1^2 + \sigma_2^2 - 2\rho
  \sigma_1\sigma_2}}\quad\mbox{and}\quad\bar v=\frac{\sigma_2(\sigma_2 -
  \rho\sigma_1)}{\sigma_1^2 + \sigma_2^2 -2\rho\sigma_1 \sigma_2}.
$$
\end{itemize}
Therefore, the behavior of the implied volatility experiences a
qualitative change (phase transition) at 
$\rho^*= \frac{\sigma_2}{\sigma_1}$. Indeed, for $\rho<\rho^*$, the expression in formula (\ref{E:assy2}),
approximating the left wing of the implied
volatility, depends on the correlation coefficient, while for $\rho>\rho^*$ the
left wing is approximated by a correlation-independent expression (see (\ref{E:assy1})). 

We will next discuss the asymptotic behavior of the implied volatility in the exceptional case 
where $n=2$ and $\rho=\rho^{*}$. The following formula holds for the distribution density $p_T$ in the exceptional case 
(see \cite{gao2009asymptotic}):
\begin{align}
&p_T(x)\approx x^{\frac{\mu_{2,T}}{T\sigma_2^2}-1}\left(\log\frac{1}{x}\right)^{-\frac{1}{T(\sigma_1^2-\sigma_2^2)}}
\left(\log\log\frac{1}{x}\right)^{-\frac{1}{2}} \nonumber \\
&\exp\left\{-\frac{1}{2T\left(\sigma_1^2-\sigma_2^2\right)}\left[\log\left(\frac{1}{\rho^2}-1\right)
+\log\log\frac{1}{x}-\log\left(\log\left(\frac{1}{\rho^2}-1\right)+\log\log\frac{1}{x}\right)+\mu_{1,T}-\mu_{2,T}
\right]^2\right\} \nonumber \\
&\exp\left\{-\frac{\log^2\frac{1}{x}}{2T\sigma_2^2}\right\}
\label{E:assy3}
\end{align}
as $x\rightarrow 0$. Recall that we assume that $\mu=0$. Recall also that $\mu_{1,T}$ and $\mu_{2,T}$ are defined in (\ref{E:weights}). 
\begin{remark}\label{R:o1} \rm Formula (\ref{E:assy3}) can be derived from formula (B20) established at the end of the proof of part (ii) of Theorem 2.3 in \cite{gao2009asymptotic}. Note that in the present paper we assume $\sigma_1\ge\sigma_2$, while in \cite{gao2009asymptotic}, $\sigma_1\le\sigma_2$. 
\end{remark}

Set
\begin{equation}
V_{1,T}=\log\left(\frac{1}{\rho^2}-1\right)+\mu_{1,T}-\mu_{2,T}\quad\mbox{and}\quad 
V_2=\log\left(\frac{1}{\rho^2}-1\right).
\label{E:assy4}
\end{equation}
It is not hard to see using the mean value theorem that
$$
\log^2\left(V_2+\log\log\frac{1}{x}\right)-\left(\log\log\log\frac{1}{x}\right)^2=o(1)
$$
as $x\rightarrow 0$. Hence
$$
\exp\left\{-\frac{1}{2T(\sigma_1^2-\sigma_2^2)}\log^2\left(V_2+\log\log\frac{1}{x}\right)\right\}
\sim\exp\left\{-\frac{1}{2T(\sigma_1^2-\sigma_2^2)}\left(\log\log\log\frac{1}{x}\right)^2\right\}
$$
as $x\rightarrow 0$. In addition,
\begin{align*}
&\exp\left\{\frac{1}{T(\sigma_1^2-\sigma_2^2)}\left(\log\log\frac{1}{x}\right)
\left(\log\left(V_2+\log\log\frac{1}{x}\right)\right)\right\} \\
&\approx\left(\log\frac{1}{x}\right)^{\frac{1}{T(\sigma_1^2-\sigma_2^2)}}
\exp\left\{\frac{1}{T(\sigma_1^2-\sigma_2^2)}\left(\log\log\frac{1}{x}\right)\left(
\log\log\log\frac{1}{x}\right)\right\}
\end{align*}
as $x\rightarrow 0$. Therefore, (\ref{E:assy3}) implies the following estimate for the density $p_T$:
\begin{align}
&p_T(x)\approx\left(\frac{1}{x}\right)^{1-\frac{\mu_{2,T}}{T\sigma_2^2}}\left(\log\frac{1}{x}\right)^{-\frac{V_{1,T}}
{T(\sigma_1^2-\sigma_2^2)}}\left(\log\log\frac{1}{x}\right)^{\frac{V_{1,T}}
{T(\sigma_1^2-\sigma_2^2)}-\frac{1}{2}} \nonumber \\
&\exp\left\{-\frac{\log^2\frac{1}{x}}{2T\sigma_2^2}\right\}
\exp\left\{-\frac{1}{2T(\sigma_1^2-\sigma_2^2)}\left(\log\log\frac{1}{x}\right)^2\right\} \nonumber \\
&\exp\left\{-\frac{1}{2T(\sigma_1^2-\sigma_2^2)}\left(\log\log\log\frac{1}{x}\right)^2\right\} \nonumber \\
&\exp\left\{\frac{1}{T(\sigma_1^2-\sigma_2^2)}\left(\log\log\frac{1}{x}\right)
\left(\log\log\log\frac{1}{x}\right)\right\}
\label{E:assy6}
\end{align}
as $x\rightarrow 0$.

Our next goal is to obtain a two-sided estimate for the put pricing function $P$, by taking into account formula (\ref{E:assy6}).
We will use the ideas employed in the proof of Theorem \ref{P:eshchio}. Let us set
$$
B(u)=\frac{u^2}{2T\sigma_2^2}+\frac{\log^2u}{2T(\sigma_1^2-\sigma_2^2)}
+\frac{(\log\log u)^2}{2T(\sigma_1^2-\sigma_2^2)}-\frac{1}{T(\sigma_1^2-\sigma_2^2)}(\log u)(\log\log u)
$$
and 
$$
a(y)=y^{-2-\frac{\mu_{2,T}}{T\sigma_2^2}}(\log y)^{-\frac{V_{1,T}}{T(\sigma_1^2-\sigma_2^2)}}
(\log\log y)^{\frac{V_{1,T}}{T(\sigma_1^2-\sigma_2^2)}-\frac{1}{2}}.
$$
It is not hard to see that the restrictions, under which formula (\ref{E:iii}) is valid, are satisfied.
In addition, for the function $b(x)=B(\log x)$, we have $b^{\prime}(x)\approx\frac{\log x}{x}$ as $x\rightarrow\infty$. Now, reasoning as in the proof of Theorem \ref{P:eshchio}, we obtain the following formula: $P(K)\approx\widetilde{P}(K)$ as $K\rightarrow 0$, where
\begin{align}
&\widetilde{P}(K)=\left(\frac{1}{K}\right)^{-1-\frac{\mu_{2,T}}{T\sigma_2^2}}\left(\log\frac{1}{K}\right)
^{-\frac{V_{1,T}}{T(\sigma_1^2-\sigma_2^2)}-2}\left(\log\log\frac{1}{K}\right)
^{\frac{V_{1,T}}{T(\sigma_1^2-\sigma_2^2)}-\frac{1}{2}} \nonumber \\
&\exp\left\{-\frac{\log^2\frac{1}{K}}{2T\sigma_2^2}\right\}
\exp\left\{-\frac{1}{2T(\sigma_1^2-\sigma_2^2)}\left(\log\log\frac{1}{K}\right)^2\right\} \nonumber \\
&\exp\left\{-\frac{1}{2T(\sigma_1^2-\sigma_2^2)}\left(\log\log\log\frac{1}{K}\right)^2\right\} \nonumber \\
&\exp\left\{\frac{1}{T(\sigma_1^2-\sigma_2^2)}\left(\log\log\frac{1}{K}\right)
\left(\log\log\log\frac{1}{K}\right)\right\}
\label{E:assy7}
\end{align}
as $K\rightarrow 0$.
Next, using (\ref{E:assy8}) 
with $\widetilde{P}$ given by (\ref{E:assy7}), and making numerous simplifications, we 
obtain the following asymptotic formula for the implied volatility in the exceptional case:
\begin{equation}
I(K)=\sigma_2+O\left(\left(\log\frac{1}{K}\right)^{-1}\right)
\label{E:assy9}
\end{equation}
as $K\rightarrow 0$. Comparing formula (\ref{E:assy9}) with formulas (\ref{E:assy1}) and (\ref{E:assy2}), we see that the behavior
of the implied volatility at the critical point $\rho=\frac{\sigma_2}{\sigma_1}$, where the 
qualitative change happens, is similar to that in the case where $\rho>\frac{\sigma_2}{\sigma_1}$.

\section{Time-changed multidimensional Black-Scholes model}\label{S:timechange}
Recall that in Section \ref{S:nBS}, we introduced the price process $S$ for a basket of assets (see formula (\ref{E:pricep})). The present section deals with time changes in such processes.
Suppose $\tau_t$, $t\ge 0$, is a non-negative non-decreasing 
stochastic process on $(\Omega,{\cal F},\{{\cal F}\}_{t\ge 0},\mathbb{P})$ (a time change). Then, the time-changed process $S$ has the following form: $t\mapsto S_{\tau_t}$. We only consider time changes which are independent of the price process $S$. In the next subsections, two-sided estimates for marginal distribution functions of time-changed price processes such as above will be established. Moreover, the leading term in the asymptotic expansion of the implied volatility associated with a time-changed price process $t\mapsto S_{\tau_t}$ in the $n$-dimensional Black-Scholes model will be found.
\subsection{Bounds on distribution functions of sums of log-normal mixtures}
The next assertion provides an upper bound for the distribution
function of a random variable imitating the random variable
$S_{\tau_t}$ for fixed $t> 0$. The additional drift vector $\tilde
\mu$ will be needed later to ensure the martingale property. 
\begin{theorem}[Upper bound]\label{T:nel}
Let $Y$ be a centered Gaussian vector with covariance matrix $\cm=[b_{ij}]_{1\le i,j\le n}$,
and let $\mu \in \mathbb R^n$ and $\tilde \mu \in \mathbb R^n$. Suppose $Z$ is a random variable with
values in $(0,\infty)$, which has a density $\rho(x)$ satisfying $\rho(s)\leq
c s^\alpha e^{-\theta s}$ for $s\geq 1$, where $\theta
>0$, $c>0$ and $\alpha \in \mathbb R$ are constants. Then, there exists $C> 0$ such that as $k\to +\infty$,
$$
\mathbb P[\sum_{i=1}^n e^{Y_i \sqrt{Z} + \mu_i Z + \tilde \mu_i} \leq e^{-k}]
\lesssim C k^\alpha e^{-c^*k },
$$
where 
\begin{equation}
c^* = \min_{t\geq 0} \max_{w\in \Delta_n} \left\{\theta t +
  \frac{(1+t\mu^\perp w)^2}{2w^\perp \cm w t}\right\}.
\label{E:prev}
\end{equation}
\end{theorem}
\begin{proof} In this proof, $C$ denotes a constant which may change
from line to line. 
For $k>0$, set
$$
F_t(k) =\mathbb P\left[ \sum_{i=1}^n e^{Y_i \sqrt{kt} + \mu_i k t +
    \tilde \mu_i}
  \leq e^{-k}\right]. 
$$
Fix $w\in \Delta_n$, and let $t$ be such that $
 1+t\mu^\perp
    w > 0. 
$
Then, by Jensen's inequality, 
\begin{align*}
&\mathbb P\left[ \sum_{i=1}^n e^{Y_i \sqrt{kt} + \mu_i k t+\tilde \mu_i}
  \leq e^{-k}\right]\leq \mathbb P\left[\sqrt{kt}\sum_{i=1}^n w_i Y_i
  + kt \mu^\perp w +\tilde \mu^\perp w+ \mathcal E(w)\leq -k\right]\\
& = N\left(- \frac{k+ tk \mu^\perp w +\tilde \mu^\perp w+ \mathcal E(w)}{\sqrt{w^\perp
      \cm w kt}}\right)\leq \frac{C \sqrt{t}}{(1+ t \mu^\perp w)\sqrt{k}}\exp\left\{-\frac{(k+ tk \mu^\perp w +\tilde \mu^\perp w+ \mathcal E(w))^2}{2w^\perp
      \cm w kt}\right\}\\
&  = \frac{C \sqrt{t}}{(1+ t \mu^\perp
  w)\sqrt{k}}\exp\left\{-k\frac{(1+ t \mu^\perp w )^2}{2 w^\perp
      \cm w t}\right\} \exp\left\{-\frac{(\tilde \mu^\perp w+\mathcal E(w))^2}{2 w^\perp
      \cm w kt}\right\}\\ &\qquad \qquad \times \exp\left\{-\frac{\mathcal E(w) +\tilde \mu^\perp w}{ w^\perp
      \cm w t}\right\} \exp\left\{-\frac{\mu^\perp w(\mathcal E(w) +\tilde \mu^\perp w)}{ w^\perp
      \cm w }\right\}\\
&\leq \frac{C \sqrt{t}}{(1+ t \mu^\perp
  w)\sqrt{k}}\exp\left\{-k\frac{(1+ t \mu^\perp w )^2}{2 w^\perp
      \cm w t}\right\}\exp\left\{-\frac{\tilde \mu^\perp w}{ w^\perp
      \cm w t}\right\},
\end{align*}
where $\mathcal E(w)$ is defined by (\ref{E:sas}).

Consider the following function:
$$
F(t,w) = \theta t + \frac{(1 + t\mu ^\perp  w)^2}{2 
  w^\perp \cm w t}.
$$
The following lemma establishes some properties of this function. The
proof is given in the appendix. 
\begin{lemma}\label{convex.lm}
There exists a unique couple $(\bar t,\bar w)$, with $\bar t\in
(0,\infty)$ and $\bar w \in \Delta_n$ such that 
$$
F(\bar t, \bar w) = \min_{t>0}\max_{w\in \Delta_n} F(t,w).
$$
In addition, the function 
$$
f(t) = F(t,\bar w)
$$
has a unique minimum at the point $\bar t$. 
\end{lemma}
We clearly have $ 1 + \bar t \mu^\perp \bar w >0$. Indeed, if $1 + \bar t \mu^\perp \bar w<0$ then $f(-\frac{1}{\mu^\perp \bar
  w})  <  f(\bar t)$ which contradicts the fact that
$\bar t$ is the minimizer. If $1 + \bar t \mu^\perp \bar w=0$ then
$f'(\bar t) = \theta$ which also leads to a contradiction.
Let 
$$
T' = \left\{\begin{aligned}
-&\frac{1}{\mu^\perp \bar w},\quad &&\mu^\perp \bar w < 0\\
&+\infty\quad &&\text{otherwise},
\end{aligned}\right.
$$
Remark that if $T'<\infty$, then
$f(T') = \theta T' > f(\bar t)$. Let us also choose $T$ small enough so that 
$$
1- |\mu^\perp \bar w| T \geq
\frac{1}{2}\quad\mbox{and}\quad\frac{1}{8\bar w^\perp \cm \bar w T} >
f(\bar t).
$$
and assume that $k$ is large enough so that $k+ 8 \tilde \mu \bar
w>0$.  
We bound the distribution function of the Gaussian mixture from above as
follows:
\begin{align}
&\mathbb P[\sum_{i=1}^n e^{Y_i \sqrt{Z} + \mu_i Z + \tilde \mu_i} \leq e^{-k}] =
\mathbb E[F_{Z/k}(k)] = \int_0^\infty \rho(s) F_{s/k}(k) ds = k
\int_0^\infty \rho(tk) F_t(k) dt\\
&\leq  k \max_{0\leq t \leq T} F_t(k)  + k \int_T^{T'}
\frac{C(tk)^\alpha\sqrt{t}}{\sqrt{k} (\mathbf 1+\mu t)^\perp w} e^{-kf(t)} dt +
c k\int_{T'}^\infty e^{-tk\theta} (tk)^{\alpha} dt. 
\label{E:nel}
\end{align} 
Now, by the choice of $T$, the first term on the right-hand side of the last inequality in (\ref{E:nel}) 
satisfies
$$
k \max_{0\leq t \leq T} F_t(k) \leq C \sqrt{k} e^{-\beta k}
$$
with $\beta > f(t^*)$. 
The second term is computed using Laplace's
method. As $k\to +\infty$, up to a constant, 
$$
k \int_T^{T'}
\frac{C(tk)^\alpha\sqrt{t}}{\sqrt{k} (\mathbf 1+\mu t)^\perp w} e^{-kf(t)} dt
\sim C k^\alpha e^{-kf(t^*)}. 
$$
Finally, the last term is negligible by the choice of $T'$.

The proof of Theorem \ref{T:nel} is thus completed.
\end{proof}

Our next goal is to establish a lower estimate complementing the estimate in Theorem \ref{T:nel}. Note that 
the estimates in Theorems \ref{T:nel} and \ref{T:nell} are off by the factor $k^{-n}$.
\begin{theorem}[Lower bound]\label{T:nell}
Let $Y$ be a centered Gaussian vector with covariance matrix $\cm$
and let $\mu \in \mathbb R^n$ and $\tilde\mu \in \mathbb R^n$. Let $Z$ be a random variable with
values in $(0,\infty)$, which has a density $\rho(x)$ satisfying $\rho(s)\geq
c s^\alpha e^{-\theta s}$ for $s\geq 1$, where $\theta
>0$, $c>0$ and $\alpha \in \mathbb R$ are constants. Then, there exists $C>0$ such that as $k\to +\infty$,
$$
\mathbb P[\sum_{i=1}^n e^{Y_i \sqrt{Z} + \mu_i Z + \tilde \mu_i} \leq e^{-k}] \gtrsim
C k^{\alpha-n}e^{-c^* k},
$$
where $c^*$ is given by (\ref{E:prev}). 
\end{theorem}
\begin{proof}
It is clear that
$$
\mathbb P\left[ \sum_{i=1}^n e^{Y_i \sqrt{kt} + \mu_i k t + \tilde \mu_i}
  \leq e^{-k}\right]\geq \mathbb P[Y_i \sqrt{kt} + \mu_i k t + \tilde \mu_i\leq - k
- \log n, i=1,\dots,n].
$$
By Proposition 3.2 in \cite{hashorva2003multivariate}, the above
probability can be bounded from below (very roughly) as follows:
\begin{align*}
\mathbb P[Y_i \sqrt{kt} + \mu_i k t + \tilde \mu_i \leq - k
- \log n, i=1,\dots,n] \geq \frac{C}{(1+k(1+t))^n} \exp\left\{-\alpha_t / 2\right\},
\end{align*}
where 
\begin{align*}
\alpha_t &= \min_{x\geq \frac{1}{\sqrt{kt}}((k+ \log n) \mathbf 1+ kt
  \mu + \tilde \mu)}x^\perp \cm^{-1}x\\
& = \max_{u\in \mathbb R^n_+} \left\{-\frac{1}{2} u^\perp \cm u +
  u^\perp \frac{1}{\sqrt{kt}}((k+ \log n) \mathbf 1+ kt
  \mu + \tilde \mu)\right\}\\
& = \max_{w\in \Delta_n} \frac{(k+\log n + kt \mu^\perp w + \tilde
  \mu^\perp w)^2}{2
  w^\perp \cm w kt}\\
&\leq \max_{w\in \Delta_n} k\frac{(1+ t \mu^\perp w)^2}{2
  w^\perp \cm w t}+ \max_{w\in \Delta_n}\frac{(1+ t \mu^\perp w)(\log
  n + \tilde \mu^\perp w)}{
  w^\perp \cm w t} + \max_{w\in \Delta_n}\frac{(\log n + \tilde
  \mu^\perp w)^2}{2
  w^\perp \cm w kt}.
\end{align*}
Finally, we bound the distribution function of the Gaussian mixture
from below as follows:
\begin{align*}
&\mathbb P[\sum_{i=1}^n e^{Y_i \sqrt{Z} + \mu_i Z + \tilde \mu_i} \leq e^{-k}] = k
\int_0^\infty \rho(tk) F_t(k) dt\geq c k \int_{\bar t - 1/k}^{\bar t+1/k}
(tk)^\alpha e^{-\theta t k} F_t(k) dt\\
&\geq \frac{Ck (\bar t k)^\alpha}{(1+k(1+t))^n}\int_{\bar t - 1/k}^{\bar t+1/k}
\exp\left\{-\theta \bar t k - k\max_{w\in \Delta_n} \frac{(1+ t \mu^\perp w)^2}{2
  w^\perp \cm w t}\right\}dt\\
&\geq \frac{C (\bar tk)^\alpha}{(1+k(1+\bar t))^n}
\exp\left\{-\theta \bar t k - k\max_{w\in \Delta_n} \frac{(1+ \bar t \mu^\perp w)^2}{2
  w^\perp \cm w \bar t}\right\} = \frac{C k^\alpha e^{-k f(\bar t)}}{(1+k(1+\bar t))^n}.
\end{align*}
\end{proof}
\begin{remark}
Theorems \ref{T:nel} and \ref{T:nell} show that under their assumptions, the dominating factor describing the decay of the left tail of
the price of a portfolio of assets is exponential with the
decay rate equal to the constant $c^*$. For example, for $n=1$, we
have
$$
c^* = \min_{t\geq 0}\{\theta t + \frac{(1+\mu t)^2}{2\sigma^2}\} =
\frac{\sqrt{2\theta \sigma^2 + \mu^2} + \mu}{\sigma^2}. 
$$
In symmetric models with $\mu = 0$, the formula for $c^*$ simplifies
to
$$
c^* = \sqrt{\frac{2\theta}{\min_{w\in \Delta_n} w^\perp \cm w}}. 
$$
\end{remark}
\subsection{Implied volatility asymptotics}\label{SS:iva}
Let $S^1,\dots,S^n$ be assets such that
$$
\log \widetilde S_t = \log \widetilde S_0  + \tilde\mu t + \mu \tau_t
+ \cm^{\frac{1}{2}} W_{\tau_t}, 
$$
where we use the same notation as in the beginning of Section
\ref{S:nBS}. Let $S$ denote the price process of the basket. Fix a
maturity $T> 0$, and
suppose the random variable $\tau_T$ has a density  $\rho_T$. Suppose
also that there exist $c_1> 0$, $c_2> 0$, $\theta> 0$ and
$\alpha\in\mathbb R$ such that 
\begin{equation}
c_1 s^\alpha e^{-\theta s}\le\rho_T(s)\le
c_2s^\alpha e^{-\theta s},\quad s\geq 1.
\label{E:timch}
\end{equation}
We assume that for every $i=1,\dots,n$, 
\begin{align}
\theta > \mu_i + \frac{\cm_{ii}}{2}. \label{mart1}
\end{align}
This assumption implies that there exists $\varepsilon>0$ such that
$$
\mathbb E[(S^i_T)^{1+\varepsilon}]<\infty
$$
We then assume further that $\tilde \mu_i$ is chosen in such way that 
\begin{align}
\mathbb E[S^i_T] = S^i_0. \label{mart2}
\end{align}

It follows from Theorems \ref{T:nel} and \ref{T:nell} 
that there exist $C_1> 0$, $C_2> 0$, and $y_0> 0$ such that
\begin{equation}
C_1y^{c^{*}}\left[\log\frac{1}{y}\right]^{\alpha-n}\le \mathbb
P[S_{\tau_T} \leq y]\le C_2y^{c^{*}} \left[\log\frac{1}{y}\right]^{\alpha},\quad y< y_0.
\label{E:twos}
\end{equation}

Since we have
$$
P(K)=\mathbb{E}\left[\left(K-S_{\tau_T}\right)^{+}\right]=\int_0^K\mathbb P[S_{\tau_T} \leq y]dy,
$$
the estimates in (\ref{E:twos}) imply that there exist $C_3> 0$, $C_4> 0$, and $K_0> 0$ such that
\begin{equation}
C_3\,K^{c^{*}+1}\left[\log\frac{1}{K}\right]^{\alpha-n}\le P(K)\le C_4\,K^{c^{*}+1}\left[\log\frac{1}{K}\right]^{\alpha},\quad K< K_0.
\label{E:twosi}
\end{equation}
Note that the put pricing pricing in (\ref{E:twosi}) is squeezed between two regularly varying functions with
the same index of regular variation at zero. Such estimates allow one to find the leading term in the asymptotic expansion of the implied volatility near zero. 
\begin{theorem}\label{T:allowim}
Suppose condition (\ref{E:timch}) holds for the time-change process
$\tau$ and that the assumptions \eqref{mart1} and \eqref{mart2} are satisfied. Then the following asymptotic formula holds for the implied volatility in time-changed $n$-dimensional Black-Scholes model:
$$
I(K)\sim\left(\frac{\psi(c^{*})}{T}\right)^{\frac{1}{2}}\sqrt{\log\frac{1}{K}}
$$
as $K\rightarrow 0$, where the function $\psi$ is defined by
\begin{equation}
\psi(u)=2-4(\sqrt{u^2+u}-u),\quad u> 0
\label{E:psik}
\end{equation}
and the constant $c^*$ is given by Formula \eqref{E:prev}.
\end{theorem} 
\begin{proof}
Theorem \ref{T:allowim} follows from (\ref{E:twosi}) and Theorem 10.28 in \cite{gulisashvili.12b}.
\end{proof}
\begin{remark}
Condition \eqref{E:timch} holds for many processes commonly used
as stochastic time changes, e.g., for the gamma process, the inverse Gaussian process,
or the generalized inverse Gaussian process. The latter process is used as time change
in the generalized hyperbolic L\'evy model. Recall that the density of the gamma process is given by
$$
\rho_t(s) = \frac{\lambda^{ct}}{\Gamma(ct)} s^{ct-1}e^{-\lambda s},
$$
while the density of the inverse Gaussian process is as follows:
$$
\rho_t(s) = \frac{ct}{s^{3/2}} e^{2ct \sqrt{\pi\lambda} - \lambda s -
  \pi c^2 t^2/s}. 
$$
In the previous formulas, the symbols $\lambda$ and $c$ stand for the parameters of the distributions. 
\end{remark}

We close this section with a counterpart of Theorem \ref{T:allowim}
for the right tail, which can be deduced from Theorem
\ref{rightcop.thm} proved in the next section. 
\begin{theorem}
Suppose condition (\ref{E:timch}) holds for the time-change process
$\tau$ and that the assumptions \eqref{mart1} and \eqref{mart2} are satisfied. Then the following asymptotic formula holds for the implied volatility in time-changed $n$-dimensional Black-Scholes model:
$$
I(K)\sim\left(\frac{\psi(c^{min})}{T}\right)^{\frac{1}{2}}\sqrt{\log{K}}
$$
as $K\rightarrow +\infty$,
where 
$$
c^{min} = \min_{i=1,\dots,n} \frac{\sqrt{2\theta \cm_{ii} + \mu_i^2} -
\mu_i}{\cm_{ii}}.
$$
\end{theorem}
\begin{proof}
Let $\overline G_i(x) = \mathbb P[\log S^i_T \geq x]$. By Theorems
\ref{T:nel} and \ref{T:nell}, there exist constants $C_1$ and $C_2$
such that
$$
C_1 x^{\alpha} e^{-c_i x} \gtrsim \overline G_i(x) \gtrsim C_2 x^{\alpha-n} e^{-c_i x}
$$
as $x\to +\infty$, where 
$$
c_i = \frac{\sqrt{2\theta \cm_{ii} + \mu_i^2} -
\mu_i}{\cm_{ii}}.
$$
Note that in the single-asset case Theorems
\ref{T:nel} and \ref{T:nell} can also be applied to the right tail, by
symmetry. It follows that 
$$
\overline G_i(x) \sim - c_i x
$$
as $x\to +\infty$, and the proof may be completed by applying Theorem
\ref{rightcop.thm}. 
\end{proof}
\section{Assets with dependence structure defined by a copula}\label{S:copula}
A popular approach to pricing European style multi-asset options is
to calibrate full-fledged models for marginal distributions of asset prices, and
then use a copula function from a simple parametric family to model the
dependence structure. This is because information about the marginal
distributions can be extracted from the prices of single asset
options, which are liquidly traded, but the market quotes
offer very little information about the dependence. 

\subsection{A very brief primer on copulas}\label{SS:primer}
Recall that the copula of a random vector
$(X_1,\dots,X_n)$ is a function $C:[0,1]^n\mapsto[0,1]$, satisfying the following
conditions:
\begin{itemize}
\item $dC$ is a positive measure in the sense of Lebesgue-Stieltjes integration.
\item $C(u_1,\dots,u_n) = 0$ when $u_k=0$ for at least one $k$. 
\item $C(u_1,\dots,u_n) = u_k$ when $u_i = 1$ for all $i\neq k$.
\end{itemize}
In addition, it is supposed that
$$
\mathbb P[X_1\leq x_1,\dots,X_n \leq x_n] = C(\mathbb P[X_1\leq
x_1],\dots,\mathbb P[X_n\leq
x_n] ),\quad (x_1,\dots,x_n)\in \mathbb R^n.
$$
A copula exists by Sklar's theorem and is uniquely defined in the case where
the marginal distributions of $X_1,\dots,X_n$ are continuous. We refer to
\cite{nelsen} for more details on copulas. 

The Gaussian copula with correlation matrix $R$ is the unique copula of any
Gaussian vector with correlation matrix $R$ and nonconstant components (it does not depend on the
mean vector and on the variances of the components).

Given a function $\phi: [0,1]\to [0,\infty]$
which is continuous, strictly decreasing and such that its inverse
$\phi^{-1}$ is completely monotonic, the Archimedean copula with generator $\phi$ is
defined by
$$
C(u_1,\dots,u_n) = \phi^{-1}(\phi(u_1)+\dots+\phi(u_n)). 
$$ 
\begin{definition}\label{wltdf.def}
The \emph{weak lower tail
  dependence function} $\chi(\alpha_1,\dots,\alpha_n)$ of a copula $C$
is defined by
$$
\chi(\alpha_1,\dots,\alpha_n) = \lim_{u\to
  0}\frac{\min_i\log u^{\alpha_i}}{\log C(u^{\alpha_1},\dots,u^{\alpha_n})},
$$
provided that the limit exists and is finite for all $\alpha_1,\dots,\alpha_n \geq 0$ such
that $\alpha_k>0$ for at least one $k$. 
\end{definition}

We will next formulate several known assertions (see \cite{tankov.14}).
\begin{theorem}\label{main.thm}
Let $X_1,\dots,X_n$ be random variables with state space $(0,\infty)$,
marginal distribution functions $F_1,\dots,F_n$, and a copula
$C$. Suppose that for every $k=1,\dots,n$, the function $F_k$ is slowly varying at
zero, and there exist constants $\eta_k$, $1\le k\le n$, and a function $F$ such that
$$
\log F_k(x) \sim \eta_k \log F(x),\quad 1\le k\le n.
$$
Suppose also that the copula $C$ admits a weak lower tail dependence function
$\chi$. 
Then,
$$
\lim_{x\downarrow 0} \frac{\log \mathbb P[X_1+\dots+X_n \leq x]}{\min_i\log \mathbb P[X_i
  \leq x]} =
\frac{1}{\chi(\eta_1,\dots,\eta_n)}. 
$$
\end{theorem}
\begin{theorem}\label{wltdf.prop}${}$
\begin{itemize}
\item Assume that a copula function $C$ has strong tail
dependence in the left tail, meaning that the limit
$$
\lambda_L = \lim_{u\downarrow 0} \frac{C(u,\dots,u)}{u},
$$ 
exists and satisfies
$\lambda_L>0$. Then, the weak lower tail dependence function of $C$
satisfies $\chi(\alpha_1,\dots,\alpha_n) = 1$.
\item Let $C$ be a Gaussian copula with correlation matrix $R$ such that $\det R
\neq 0$. Then,
$$
\chi(\alpha_1,\dots,\alpha_n) = \max_i \alpha_i \min_{w\in \Delta_n}
w^T \Sigma w,\quad \text{for all $\alpha_1,\dots,\alpha_n>0$,}
$$
where the matrix $\Sigma$ has entries $\Sigma_{ij} =
\frac{R_{ij}}{\sqrt{\alpha_i\alpha_j}}$, $1\leq i,j\leq n$. 
\item Let $C$ be an Archimedean copula with a generator function $\phi$
  such that $\log \phi^{-1}$ is regularly varying at $\infty$ with index
$\lambda>0$. Then,
$$
\chi(\alpha_1,\dots,\alpha_n) = \frac{\max(\alpha_1,\dots,\alpha_n)}{(\alpha_1^{1/\lambda}+\dots + \alpha_n^{1/\lambda})^\lambda}.
$$
\end{itemize}
\end{theorem}

\subsection{Copulas and the implied volatility asymptotics}\label{SS:civa}
In this subsection, we study the left-wing behavior of the
implied volatility associated with a basket call option. Recall that we denoted by $(Y_1,\dots,Y_n)$ the vector of logarithmic returns of the risky assets, and by $(\lambda_1,\dots,\lambda_n)$ the corresponding vector of weights. Let $C$ be the copula of the
vector $(Y_1,\dots,Y_n)$, and $G_i$ be the distribution function of $Y_i$ for $i=1,\dots,n$. The implied volatility is considered in this section as a function $k\mapsto I(-k)$ of the variable $-k$, where $k$ is the log-strike defined by $k=\log K$. The tail-wing formulas due to Benaim and Friz (see \cite{benaim.friz.09}) play an important role in the sequel.
\begin{theorem}\label{tailwing.cop}
Let $\alpha>0$, and assume that the following are true: 
\begin{itemize}
\item There exists $\varepsilon>0$ such that $\mathbb
  E[e^{-\varepsilon Y_i}]<\infty$, $i=1,\dots,n$. 
\item For every $1\le i\le n$, the function $k\mapsto-\log G_i(-k)$, $k> k_0$, belongs to the class $R_\alpha$ of regularly varying functions, and there exist positive constants $\eta_1,\dots,\eta_n$ and a function $G$ such that 
\begin{align}
\log G_i(-k)\sim \eta_i \log G(-k)\quad \text{as $k\to\infty$.} \label{samelaw.eq}
\end{align}

\item The copula $C$ admits a weak lower tail dependence function
  $\chi$. 
\end{itemize}
Then,
\begin{equation}
\frac{I(-k)^2 T}{k}\sim \psi\left[-\frac{\log G(-k)}{k}\frac{\max_i
  \eta_i }{\chi(\eta_1,\dots,\eta_n)}\right]
\label{E:comp}
\end{equation}
as $k\to\infty$, where the function $\psi$ is defined in (\ref{E:psik}). 
\end{theorem}
 \begin{proof}
The distribution function $F_i$ of the random variable $\lambda_i S_i$ is given by 
$$
F_i(x) = G_i(\log x - \log \lambda_i).
$$ 
Since the function $\log G_i$ is regularly varying at
$-\infty$, it is clear that $\log F_i$ is slowly varying at zero and
$$
\log F_i(x)\sim \log G_i(\log x)\sim \eta_i \log G(\log x)
$$ 
as $x\to 0$. It follows from Theorem \ref{main.thm} that 
$$
\log F(x) \sim \frac{\max_i \eta_i}{\chi(\eta_1,\dots,\eta_n)
} \log G(\log x)\quad \text{as $x\to 0$},
$$
where $F$ is the distribution function of $\sum_{i=1}^n \lambda_i S_i$. Equivalently
$$
\log F(e^{-k}) \sim \frac{\max_i \eta_i}{\chi(\eta_1,\dots,\eta_n)
} \log G(-k)\quad \text{as $k\to \infty$},
$$
and hence
\begin{equation}
-\frac{\log F(e^{-k})}{k} \sim-\frac{\log G(-k)}{k}\frac{\max_i \eta_i}{\chi(\eta_1,\dots,\eta_n)} 
\quad \text{as $k\to \infty$}.
\label{E:ere}
\end{equation}
It follows from the assumptions in Theorem \ref{tailwing.cop} that $\log G(-k)\in
R_\alpha$ as $k\to\infty$. Therefore $\log F(e^{-k}) \in
R_\alpha$ as well. Next, using the tail-wing formula of Benaim and Friz (see Theorem 2 in
\cite{benaim.friz.09}), we obtain 
\begin{equation}
\frac{I(-k)^2 T}{k}\sim \psi\left[-\frac{\log F(e^{-k})}{k}\right]\quad
\text{as $k\to\infty$.} 
\label{E:willn}
\end{equation}

We will need the following lemma.
\begin{lemma}\label{L:lim}
Let $\psi$ be the function defined by (\ref{E:psik}), and suppose $\rho_1$ and $\rho_2$ are positive functions on $(0,\infty)$ such that 
\begin{equation}
\frac{\rho_1(x)}{\rho_2(x)}\rightarrow 1\quad\mbox{as}\quad x\rightarrow\infty.
\label{E:pro1}
\end{equation} 
Then
\begin{equation}
\frac{\psi(\rho_1(x))}{\psi(\rho_2(x))}\rightarrow 1\quad\mbox{as}\quad x\rightarrow\infty.
\label{E:pro2}
\end{equation} 
\end{lemma}
\begin{proof}
It is not hard to see that for all $u\ge 0$,
\begin{equation}
\psi(u)=\frac{2}{(\sqrt{u+1}+\sqrt{u})^2}.
\label{E:use}
\end{equation}
The equality in (\ref{E:use}) describes the structure of the function $\psi$ better than the original definition.

Fix $\varepsilon> 0$. Then, using (\ref{E:use}) and the inequality $1<\frac{1}{1-\varepsilon}$, we get
\begin{align*}
\psi((1-\varepsilon)u)&\le\frac{2}{(1-\varepsilon)\left(\sqrt{u+\frac{1}{1-\varepsilon}}+\sqrt{u}\right)^2}
\le\frac{2}{(1-\varepsilon)(\sqrt{u+1}+\sqrt{u})^2} \\
&=\frac{2}{(1-\varepsilon)(\sqrt{u+1}+\sqrt{u})^2}=
\frac{1}{1-\varepsilon}\psi(u).
\end{align*}
Similarly
$$
\psi((1+\varepsilon)u)\ge\frac{1}{1+\varepsilon}\psi(u).
$$
Therefore,
\begin{equation}
\frac{1}{1+\varepsilon}\psi(u)\le\psi((1+\varepsilon)u)\le\psi((1-\varepsilon)u)\le\frac{1}{1-\varepsilon}\psi(u).
\label{E:ost}
\end{equation}

It follows from (\ref{E:pro1}) that for every $\varepsilon> 0$ there exists $x_{\varepsilon}> 0$ such that
$$
(1-\varepsilon)\rho_2(x)\le\rho_1(x)\le(1+\varepsilon)\rho_2(x)
$$
for all $x> x_{\varepsilon}$. Since the function $\psi$ decreases on $(0,\infty)$, we have
$$
\psi((1+\varepsilon)\rho_2(x))\le\psi(\rho_1(x))\le\psi((1-\varepsilon)\rho_2(x))
$$
for all $x> x_{\varepsilon}$. Now, using (\ref{E:ost}), we obtain
$$
\frac{1}{1+\varepsilon}\psi(\rho_2(x))\le\psi(\rho_1(x))\le\frac{1}{1-\varepsilon}\psi(\rho_2(x))
$$
for all $x> x_{\varepsilon}$, and (\ref{E:pro2}) follows.
\end{proof}
Finally, it is not hard to see that (\ref{E:ere}), (\ref{E:willn}), and Lemma \ref{L:lim}
imply (\ref{E:comp}).

This completes the proof of Theorem \ref{tailwing.cop}.
\end{proof}

The next example shows that condition \eqref{samelaw.eq} does not
prevent one from choosing different marginal laws for different
components of the process $(Y_1,\dots,Y_n)$ as long as these laws have a similar
tail behavior. 
\begin{example}
Let us consider the following multidimensional extension of the example given in
Section 5.2 of \cite{benaim.friz.09}. We assume that for $i=1,\dots,n$,
the distribution of the random variable $Y_i$ is normal inverse Gaussian, more precisely,
NIG$(\alpha_i,\beta_i,\mu_i,\delta_i)$. It is also supposed that the
parameters satisfy $\alpha_i>|\beta_i|>0$ and $\delta_i>0$. This means that the moment
generating function of $Y_i$ is given by 
$$
M_i(z) = \exp\left(\delta_i\left\{\sqrt{\alpha_i^2 - \beta_i^2} -
    \sqrt{\alpha_i^2 - (\beta_i+z)^2}\right\}+\mu_i z\right).
$$ 
We refer the reader to \cite{bns_nig} for more details on the normal inverse
Gaussian distribution. In particular, it follows that $Y_i$ has a
density $g_i$ which satisfies the following condition:
$$
g_i(k)\sim C_i |k|^{-\frac{3}{2}} e^{-\alpha_i |k| + \beta_i k},\quad k\to
\pm \infty,
$$ 
where $C_i$ is a constant. Using Theorem 2 in \cite{benaim.friz.09}, we
see that $-\log G_i(-k)\in R_\alpha$ as $k\to +\infty$, and also 
$$
-\log G_i(-k) \sim -\log g_i(-k) \sim (\beta_i-\alpha_i)k,\quad k\to
+\infty. 
$$ 
Therefore, the condition in \eqref{samelaw.eq} holds with $\lambda_i
= \alpha_i-\beta_i$ and $G(k) = e^k$. 

Assuming that the dependence structure of $(Y_1,\dots,Y_n)$ is described by the
Gaussian copula with correlation matrix $R$, we see that 
\begin{align}
\frac{I(-k)^2 T}{k} \sim \psi\left[\frac{1}{\inf_{w\in \Delta_d} w^\perp
  \Sigma w}\right],\quad k\to +\infty,\label{nigres.eq}
\end{align}
where the matrix $\Sigma=[\Sigma_{ij}]$ is such that 
$$
\Sigma_{ij} = \frac{R_{ij}}{\sqrt{(\alpha_i-\beta_i)(\alpha_j-\beta_j)}}.
$$
In other words, the implied variance is asymptotically linear, with a correlation-dependent
limiting slope, which is given by the right-hand side of \eqref{nigres.eq}. 
\end{example}

For the sake of completeness, we include a proposition that is a counterpart of Theorem
\ref{tailwing.cop} in the case of the right tail. This proposition turns out to be somewhat
trivial: the leading order of the implied volatility is determined by
a single component with the fattest tail, and it does not depend on the
copula. Let us denote by $\overline G_i$ the survival function of $Y_i$, i.e.,
the function $\overline G_i(x) =\mathbb P[Y_i \geq x]$. 
\begin{theorem}\label{rightcop.thm}
Let $\alpha>0$, and suppose that the following assumptions hold:
\begin{itemize}
\item There exists $\varepsilon>0$ such that $\mathbb
  E[e^{(1+\varepsilon)Y_i}]<\infty$ for $i=1,\dots,n$.
\item For each $i=1,\dots,n$, the function $k\mapsto-\log \overline G_i(k)$ belongs to the class $R_\alpha$ at
 infinity. 
\end{itemize}
Then,
\begin{equation}
\frac{I(k)^2 T}{k} \sim \psi\left[-\frac{1}{k}\max_i \log \overline
  G_i(k)\right]\quad \text{as $k\to+\infty$.}
\label{E:finally}
\end{equation}
\end{theorem}
\begin{proof}
Set $X_i = v_i e^{Y_i}$. Then we get
\begin{align*}
&\mathbb P[X_1+\dots+X_n \geq x] \geq \max_i \mathbb P[X_i \geq x],\\
&\mathbb P[X_1+\dots+X_n \geq x]\leq \mathbb P[\exists i: X_i \geq
\frac{x}{n}] \leq \sum_{i=1}^n \mathbb P[ X_i \geq
\frac{x}{n}]\leq n \max_i \mathbb P[ X_i \geq
\frac{x}{n}].
\end{align*}
Since for each $i$, the function $\log\overline G_i$ is regularly varying at infinity, it follows
that the function $x\mapsto\log \mathbb P[X_i \geq x]$ is slowly varying, and therefore, for $x$ sufficiently large and any $\varepsilon>0$,
$$
\max_i\log \mathbb P[X_i \geq x/n] \leq (1+\varepsilon) \max_i \log \mathbb P[X_i \geq x].
$$
Finally,
$$
\lim_{x\to +\infty}\frac{\log \mathbb P[X_1+\dots+X_n \geq x]}{\max_i
  \log \mathbb P[X_i \geq x]} = 1,
$$
and formula (\ref{E:finally}) follows from Theorem 1 in
\cite{benaim.friz.09} with a similar proof to that of Theorem \ref{tailwing.cop}. 
\end{proof}

\appendix
\section{Proof of Lemma \ref{convex.lm}}
The function $F$ satisfies
$$
F(t,w) = \max_{\lambda >0} \{\theta t + \lambda w^\perp (\mathbf 1+\mu t) -
\frac{\lambda^2 w^\perp \cm w t}{2}\},
$$
where $\mathbf 1$ stands for the $n$-dimensional vector with all
elements equal to $1$.
Therefore,
$$
\max_{w\in \Delta_n} F(t,w) = \max_{u \in \mathbb R_+^n } \widetilde
F(t,u),
$$
with
$$
\widetilde
F(t,u) = \{\theta t + u^\perp (\mathbf 1+\mu t) -
\frac{u^\perp \cm u t}{2}\}.
$$
Since for every $t>0$, $\widetilde F(t,u)$ is strictly concave in $u$,
there exists a unique $\bar u(t)\in \mathbb R^n_+$ with $\bar u(t)\neq
0$ such that $\widetilde F(t,\bar u) =
\max_{u \in \mathbb R_+^n } \widetilde F(t,u)$. This in turn implies
that there exists a unique $\bar w(t)$ such that $F(t,\bar w) =
\max_{w\in \Delta_n } F(t,w)$. It is also easy to see that $\bar u(t)$
depends continuously on $t$. 

Let $\bar f(t) =\widetilde
F(t,\bar u(t))$. We would like to show that $\bar f$ is differentiable
in $t$ and compute its derivative. $\bar u(t)$
may be characterized as follows: for $i=1,\dots,n$
\begin{align}
&[\mathbf 1+ \mu t - t\cm \bar u(t) ]_i = 0\quad \text{if}\quad \bar u(t)_i >0\label{mincarac1}\\
&[\mathbf 1+ \mu t - t\cm \bar u (t)]_i \leq 0\quad \text{if}\quad \bar u(t)_i =0.\label{mincarac2} 
\end{align}
Let $I(t)$ denote the set of indices $i\in\{1,\dots,n\}$ such that
$\bar u(t)_i>0$, and, for a vector $x\in \mathbb R^n$, let $x_{I(t)}$
denote the subset of components of $x$ with indices in $I(t)$:
$x_{I(t)} = \{x_i: i\in I(t)\}$. Furthermore, let $\cm_{I(t),I(t)}$
denote the submatrix of the covariance matrix, containing the elements
$b_{ij}$ with $i\in I(t)$ and $j\in I(t)$. Then, the vector $\bar
u(t)$ satisfies
$$
\bar u(t)_{I(t)} = \frac{1}{t}\cm_{I(t),I(t)}^{-1} (\mathbf 1+\mu
t)_{I(t)},\quad \bar u(t)_{\tilde I(t)}  = 0,
$$
where the set $\tilde I(t)$ contains the indices $i\in\{1,\dots,n\}$
which are not in $I(t)$. 

Now, fix $t\in (0,\infty)$ and for $t'\in (0,\infty)$, define
$$
v(t')_{I(t)} = \frac{1}{t'}\cm_{I(t),I(t)}^{-1} (\mathbf 1+\mu
t')_{I(t)},\quad v(t)_{\tilde I(t)}  = 0
$$
First, assume that for all $i$ such that
$\bar u(t)_i =0$, 
either $[\mathbf 1+ \mu t - t\cm \bar u (t)]_i <0$ (with strict
inequality) or 
$$
[\mathbf 1+\mu t' - t' \cm v(t')]_i = 0
$$
for all $t'\in (0,\infty)$. We shall call this assumption Assumption
1. 
Then  we can find $\delta>0$, such that for
every $t' \in(0,\infty)$ with $|t'-t|<\delta$, $v(t')$ satisfies the characterization
\eqref{mincarac1}--\eqref{mincarac2}. Therefore, $v(t') = \bar
u(t')$. This means that 
$$
\bar f(t') = \theta t' + \frac{1}{2t'} (\mathbf 1+\mu
t')_{I(t)}^{\perp}\cm_{I(t),I(t)}^{-1} (\mathbf 1+\mu
t')_{I(t)}.
$$
Therefore, $\bar f$ is differentiable at $t$ with first derivative given by
\begin{align}\bar f'(t) = \theta - \frac{1}{2t^2}
\mathbf 1_{I(t)}^{\perp}\cm_{I(t),I(t)}^{-1} \mathbf 1_{I(t)} + \frac{1}{2} \mu
_{I(t)}^{\perp}\cm_{I(t),I(t)}^{-1} \mu
_{I(t)} = \theta - \frac{1}{2t} \bar u(t)^\perp (\mathbf 1-\mu t)\label{fderiv.eq}
\end{align}
and second derivative 
$$
\bar f^{\prime\prime}(t) = \frac{1}{t^3}
\mathbf 1_{I(t)}^{\perp}\cm_{I(t),I(t)}^{-1} \mathbf 1_{I(t)} .
$$

Now assume that there exists at least one $i$ such that $\bar u(t)_i =
0$ and $[\mathbf 1+ \mu t - t\cm \bar u (t)]_i =0$, or, equivalently, 
$$
[\mathbf 1+ \mu t' - t'\cm v (t')]_i =0
$$
with $t'=t$. The case when the above equality holds for all $t'$ is
covered by Assumption $1$. Since the left-hand side is linear in $t'$,
this means that for a given index set $I(t)$ and for a given $i$,
there exists only one $t'\in(0,\infty)$ which satisfies the above
equality. Since the number of possible index sets is finite, we
conclude that there is at most a finite number of elements $t\in
(0,\infty)$ which do not satisfy Assumption 1. But then, we can
conclude by continuity that $\bar f$ is strictly convex (which
entails uniqueness of $\bar t$) and
differentiable for all $t\in (0,\infty)$, with the derivative given by
\eqref{fderiv.eq} or alternatively by 
$$
\bar f'(t) = \theta -\frac{1}{2 t^2 \bar w(t)^\perp \cm \bar w(t)} +
\frac{(\bar w(t)^\perp \mu)^2}{2 \bar w(t)^\perp \cm \bar w(t)}. 
$$
Comparing this with the derivative of $f$, which is
easily computed, we see that at the point $\bar t$, these derivatives
coincide. Since this point is characterized by the first order
condition $\bar f'(\bar t)=0$, and the function $f$ is stictly convex,
$f$ also attains its unique minumum at $\bar t$.

\end{document}